\documentclass[11pt]{article}

\usepackage{verbatim} 

\usepackage{makeidx}
\usepackage{amsmath,amssymb,amsfonts}
\usepackage{dsfont}
\usepackage{latexsym}
\usepackage{subfigure}
\usepackage{graphicx}
\usepackage{times}\usepackage[scaled=0.92]{helvet} 
\usepackage{fancybox}
\usepackage{fullpage}
\usepackage{hyperref}
\usepackage{paralist}

\usepackage{color}
\usepackage{wrapfig}
\usepackage{tikz}

\usepackage[small,compact]{titlesec}

\newcount\shortyear\newcount\shorthour\newcount\shortminute
\shorthour=\time\divide\shorthour by 60\shortyear=\shorthour
\multiply\shortyear by 60\shortminute=\time\advance\shortminute by
-\shortyear \shortyear=\year\advance\shortyear by -1900

\def\zeit{\number\shorthour:\ifnum\shortminute<10 0\number\shortminute
	\else\number\shortminute\fi}

\newenvironment{proof}{\noindent{\bf Proof : \ }}{\hfill$\Box$\par\medskip}

\newtheorem{theorem}{Theorem}
\newtheorem{corollary}[theorem]{Corollary}
\newtheorem{lemma}[theorem]{Lemma}

\newenvironment{proofof}[1]{\begin{trivlist} \item {\bf Proof
			#1:~~}}
	{\qed\end{trivlist}}

\newcommand{\COMMENTED}[1]{{}}

\newcommand{\alg}{{\sf Alg}}

\newcommand{\prob}[1]{\operatorname{Pr}\left[#1\right]}
\newcommand{\ex}[1]{\operatorname{E}\left[#1\right]}

\newcommand{\E}{\ensuremath{\mathbf{E}}}

\renewcommand{\i}{\beta}
\renewcommand{\H}{H_{\i}}
\newcommand{\Ha}{H^{\alpha}}
\renewcommand{\S}{{\ensuremath{\mathbb{S}}}}
\newcommand{\optm}{{\ensuremath{\sf{opt}_M}}}
\newcommand{\optd}{{\ensuremath{\sf{opt}_D}}}
\newcommand{\optc}{{\ensuremath{\sf{opt}_C}}}
\newcommand{\eps}{{\ensuremath{\varepsilon}}}

\newcommand{\eat}[1]{}

\title{Metric Sublinear Algorithms via Linear Sampling }
\author{
	Hossein Esfandiari\thanks{Supported in part by NSF grants CCF-1320231 and CNS-1228598.}\\Harvard University \\ Cambridge,
	MA 
	\and
	Michael Mitzenmacher\thanks{Supported in part by  NSF grants CCF-1563710, CCF-1535795, CCF-1320231, and CNS-1228598.  Part of this work was done while visiting Microsoft Research New England.}\\Harvard University \\ Cambridge,
	MA 
}
\date{}

\begin{document}\sloppy 
\maketitle
\begin{abstract}
In this work we provide a new technique to design fast approximation algorithms for
graph problems where the points of the graph lie in a metric space.  
Specifically, we present a sampling approach for such metric graphs that, using a sublinear number of edge weight queries, provides a {\em linear sampling}, where each
edge is (roughly speaking) sampled proportionally to its weight.

For several natural problems, such as densest subgraph and max cut among others, we show that by sparsifying the graph using this sampling process, we can run
a suitable approximation algorithm on the sparsified graph and the result remains a good approximation for the original problem.
Our results have several interesting implications, such as providing the first sublinear time approximation algorithm for densest subgraph in a metric space, and improving the running time of estimating the average distance. 
\end{abstract}

\section{Introduction}
In this paper, we aim to design approximation algorithms for several natural graph problems, in the setting where the points in the graph lie in a metric
space.  
Following the seminal work of \cite{indyk1999sublinear}, we aim to provide {\em sublinear} approximation algorithms; that is, on problems with $n$
points and hence ${n \choose 2}$ edge distances, we aim to provide
randomized algorithms that require $o(n^2)$ time and in fact only consider
$o(n^2)$ edges, by making use of sampling. Similar to the previous work, we assume we can query the weight of any single edge in $O(1)$ time;
when we use the term ``query'', we mean an edge weight query throughout.  

A well known technique to design sublinear algorithms is uniform sampling;  that is, a subset of edges (or vertices) is sampled uniformly at random. Several algorithms use uniform sampling to improve speed, space, or the number of queries~\cite{ahn2012analyzing,ahn2012graph,ahn2013spectral,alon2003testing,bateni2016almost,bateni2016distributed,bhattacharya2015space,czumaj2004sublinear,esfandiari2015applications,mcgregor2015densest}. Uniform sampling is very easy to implement, but problematically it is oblivious to the edge weights. When it comes to maximization problems on graphs, a few high weight edges may have a large effect on the solution, and hence the uniform sampling technique may fail to provide a suitable solution because it fails to sample these edges. For example, consider the densest subgraph problem, where the density of a subgraph is the sum of the edges weights divided by the number of vertices.  It is known that for general unweighted graphs, the densest subgraph of a uniformly sampled subgraph with $\tilde{O}(\frac n {\epsilon^2})$ edges is a $1-\epsilon$ approximation of the densest subgraph of the original graph \cite{esfandiari2015applications,mcgregor2015densest,mitzenmacher2015scalable}.  However, as we show in Appendix~\ref{app:dens} this result is not true for weighted graphs, even in a metric space.
This problem suggests we should design approaches that sample edges with probabilities proportional to (or otherwise related to) their weight in a metric space.

As our main result, we design a novel sampling approach using a sublinear number of queries for graphs in a metric space,
where independently for each edge, the probability the edge is in the sample is
proportional to its weight; we call such a sampling a {\em linear sampling}. Specifically, for a fixed factor $\alpha$, we can ensure
for an edge $e$ with weight $w_e$, if $\alpha w_e \leq 1$
then the edge appears in the sample with weight 1 with probability $\alpha
w_e$, and if $\alpha w_e > 1$, then the edge
is in the sample with weight $\alpha w_e$.  Hence the edge weights are ``downsampled'' by a factor of $\alpha$,
in a natural way.  
%
%
We can choose an $\alpha$ to suitably sparsify our sample,
graph, run an approximation algorithm on that sample, and use that
result to obtain a corresponding, nearly-as-good approximation to the
original problem. Interestingly, we only query $\tilde{O}(n+\i)$\footnote{$\tilde{O}(\cdot)$ notation hides logarithmic factors.} edge weights to provide the sample, where $\i$ is ``almost'' the expected weight of the edges in the sampled graph. (See Subsection \ref{subsec:results} for a formal definition). Our algorithm to construct the sample also runs in $\tilde{O}(n+\i)$ time. 

Utilizing our sampling approach, we show that for several problems a $\phi$-approximate solution on a linear sample with expected weight (roughly) $\i\in o(n^2)$ is a $(\phi-\epsilon)$-approximate solution on the input graph.
From an information theory perspective this says that $\tilde{O}(n+\i)$ queries are sufficient to find a $1-\epsilon$ approximate solution for these problems. Moreover, as the sampled graph has a reduced number of edges, if an approximation algorithm on the sampled graph runs in linear time on the sampled edges, the total time is sublinear in the size of the original graph.  

In what follows, after describing the related work and a summary of
our results, we present our sampling method.  Our approach decomposes
the graph into a sequence of subgraphs, where the decomposition
depends strongly on the fact that the graph lies in a metric space.
Using this decomposition, and an estimate of the average edge weight
in the graph, we can determine a suitable sampled graph. 
We then show this sampling approach allows us to find sublinear approximation
algorithms for several problems, including densest subgraph and max cut, in
the manner described above.  

In some applications, such as diversity maximization, it can be beneficial to go slightly beyond metric distances~\cite{zadeh2017scalable}. 
We can extend our results to more general spaces that satisfy what is commonly referred to as a parametrized triangle
inequality~\cite{andreae1995performance,bender2000performance,chandran2002approximations}, in which for every three points $a$, $b$ and $c$ we have $w_{a,b}+w_{b,c} \geq \lambda w_{c,a}$ for a parameter
$\lambda$.   As an example, if the weight of each edge $(u,v)$ is the squared distance between the two points, the graph satisfies 
a parametrized triangle inequality with $\lambda = 1/2$.  We provide analysis for this more general setting throughout, and
refer to a graph satisfying such a parametrized triangle inequality as a $\lambda$-metric graph.  (Throughout, we take $\lambda \leq 1$).


\subsection{Our Results}\label{subsec:results}
As our main technical contribution we provide an approach to sample a graph $\H = (V,E_H)$ from a $\lambda$-metric graph $G=(V,E_G)$ with the properties specified below that makes only $\tilde{O}(\frac {n+\i}{\lambda})$ queries and succeeds with probability at least $1-O(1/n)$. It is easy to observe that our algorithm runs in $\tilde{O}(\frac{n+\i}{\lambda})$ time as well.
\begin{itemize}
	\item For some fixed factor $\alpha$ (which is a function of $\beta$) independently for each edge $e$ we have:
	\begin{itemize}
		\item If $\alpha w_e\leq 1$, we have edge $e$ with weight $1$ in $E_H$ with probability $\alpha w_e$.
		\item If $\alpha w_e > 1$, we have edge $e$ with weight $\alpha w_e$ in $E_H$.
	\end{itemize}
	\item We have $\i \leq \ex{\sum_{e\in E_H} w'_e} \leq 2\i $, where $w'_e$ is the weight of $e$ in $\H$.\footnote{This can be extended to $\i \leq \ex{\sum_{e\in E_H} w'_e} \leq (1+\gamma)\i $, for any arbitrary $\gamma$ (See the footnote on Theorem \ref{thm:sketch:mainH} for details.)  In our work, the upper bound only affects the number of queries; we prefer to set $\gamma = 1$ and simplify the argument.}
\end{itemize}
As the weight of each edge in $E_H$ is at least $1$, $\ex{\sum_{e\in E_H} w_e} \leq 2\i$ implies that $\ex{|E_H|} \leq 2\i$.

We note that for three points $a$, $b$ and $c$ in a $\lambda$-metric space and any parameter $p$, $w_{a,b}+w_{b,c} \geq \lambda w_{c,a}$ directly implies $w_{a,b}^p+w_{b,c}^p \geq \frac {\lambda} {2^p} w_{c,a}^p$. Therefore one can use our technique to sample edges proportional to $w_e^p$ (a.k.a. $l_p$ sampling). In the streaming setting, $l_p$ sampling has been extensively studied and appears to have several applications~\cite{monemizadeh20101}; as far as we are aware, our approach provides the first $l_p$ sampling techniques that uses a sublinear number of edge weight queries. 

As previously mentioned, in Section~\ref{sec:app} we consider several problems and show that for some $\i \in o(n^2)$, any $\phi$-approximate solution of the problem on $\H$ is an $(\phi-\epsilon)$-approximate solution on the original graph with high probability. Specifically, we show that $\i \in O(\frac{n \log n}{\epsilon^2})$ is sufficient to approximate densest subgraph and max cut, $\i \in O(\frac{n^2 \log n}{\epsilon^2k})$ is sufficient to approximate $k$-hypermatching, and $\i \in O(\frac{\log n}{\epsilon^2})$ is sufficient to approximate the average distance. Notice that these results directly imply (potentially exponential time) $(1-\epsilon)$ approximation algorithms with sublinear number of queries for each of the problems.  Often our methodology can also yield sublinear time algorithms (since it uses a sublinear number of edges) with possibly worse approximation ratios.  

We now briefly describe specific results for the various problems we consider, although we defer the formal problem definitions to Section \ref{sec:app}. All of the algorithms discussed below work with high probability.  We note that, throughout the paper, we use $\log n$ for $\log_e n$.  

For average distance, we provide a $(1-\epsilon)$-approximation algorithm that simply finds the sum of the weights of the edges in $\H$ for $\i \in O(\frac{\log n}{\epsilon^2})$, and hence our algorithm runs in time $\tilde{O}(\frac{n+\frac{1}{\epsilon^2}}{\lambda})$. For a metric graph, this improves the running time of the previous result of Indyk~\cite{indyk1999sublinear} that runs in $O(\frac{n}{\epsilon^{3.5}})$ time, with constant probability.


For densest subgraph, the greedy algorithm yields a $1/2$-approximate solution in time quasilinear in the number of edges~\cite{charikar2000greedy}. The expected number of edges of $\H$ can be bounded by $\tilde{O}(\frac{n}{\lambda \epsilon^2})$ for the densest subgraph on $\lambda$-metric graphs. Therefore, our result implies a $(1/2-\epsilon)$-approximation algorithm for densest subgraph in $\lambda$-metric spaces requiring $\tilde{O}(\frac{n}{\lambda \epsilon^2})$ time.

A sublinear time algorithm for a $(1-\epsilon)$ approximation for metric max cut is already known~\cite{indyk1999clustering}. The previous result uses $\tilde{O}(\frac n {\epsilon^5})$ queries, while we use only $\tilde{O}(\frac n {\epsilon^2})$ queries.  (We note that this result does not improve the running time, but remains interesting from an information theoretic point of view. Indeed, there are several interesting results on sublinear space algorithms that ignore the computational complexity
e.g., max cut~\cite{bhaskara2018sublinear,kapralov20171+,kapralov2015streaming,kogan2015sketching},
set cover~\cite{assadi2017tight,har2016towards},
vertex cover and hypermatching~\cite{chitnis2015parameterized,chitnis2016kernelization}.) 
 
Finally, on the hardness side, in Section \ref{sec:hard} we show that $\Omega(n)$ queries are necessary even if one just wants to approximate the size of the solution for densest subgraph, $k$-hypermatching, max cut, and average distance. 

\subsection{Other Related Work}
Metric spaces are natural in their own right.  For example, they represent
geographic information, and hence graph problems such as the densest
subgraph problem often have a natural interpretation in metric spaces.
It also is often reasonable to manage large data sets by
embedding objects within a suitable metric space.  In networks, for
example, the idea of finding network coordinates consistent with latency
measurements to predict latency has been widely studied~\cite{cox2004practical,ledlie2007network,Pietzuch,shavitt2004big,tang2015line,verbeek2016metric}.  

There are several works on designing sublinear algorithms for
different variants of clustering problems in metric spaces due to
their application to machine
learning~\cite{alon2003testing,buadoiu2005facility,czumaj2004sublinear,czumaj2007small,indyk1999clustering}.
We briefly summarize some of these papers.  Alon et al. studies the
efficiency of uniform sampling of vertices to check for given parameters $k$
and $b$ if the set of points can be clustered into $k$ subsets each
with diameter at most $b$, ignoring up to an $\epsilon$ fraction of
the vertices~\cite{alon2003testing}.
Czumaj and Sohler studies the efficiency of uniform sampling of
vertices for $k$-median, min-sum $k$-clustering, and balanced
$k$-median~\cite{czumaj2004sublinear}.
Badoiu et al. consider the facility location problem in metric space~\cite{buadoiu2005facility}. They compute the optimal cost of the minimum facility location problem, assuming uniform costs and demands, and assuing every point can open a facility. Moreover, they show that there is no $o(n^2)$ time algorithm that approximates the optimal solution of general case of metric facility location problem to within any factor.


A basic and natural difference between these previous works on clustering problems and the densest subgraph problem that we consider here is that all previous problems aim to decompose the graph into two or more subsets, where each subset consists of points that are close to each other. However, densest subgraph in a metric space aims to pick a diverse, spread out subset of points. (While perhaps counterintuitive, this is clear from the definition, which we provide shortly.)  The application of metric densest subgraph in diversity maximization and feature selection is well studied~\cite{bhaskara2016linear,zadeh2017scalable}.

Sublinear algorithms may also refer to sublinear space algorithms such
as streaming algorithms.  A related, well-studied setting is
semi-streaming~\cite{muthukrishnan2005data}, often used for graph
problems. In the semi-streaming setting the input is a stream of edges
and we take one (or a few) passes over the stream, while only using
$\tilde{O}(n)$ space.  Semi-steaming algorithms have been extensively
studied~\cite{ahn2009,AG11,EpsteinLMS09,mcgregor2005,KelnerL11,KonradR13}.

For the densest subgraph problem, there have been a number of recent papers showing the efficiency of uniform edge sampling in unweighted graphs~\cite{bhattacharya2015space,esfandiari2015applications,mcgregor2015densest,mitzenmacher2015scalable}. Initially, Bhattacharya et al. provided a $0.5$ approximation semi-streaming algorithm for this problem~\cite{bhattacharya2015space}. They extended their approach to obtain a $0.25$ approximation algorithm for this problem for dynamic streams with $\tilde{O}(1)$ update time and $\tilde{O}(n)$ space. McGregor et al. and Esfandiari et al. independently provide a $(1-\epsilon)$-approximation semi-streaming algorithm for this problem~\cite{esfandiari2015applications,mcgregor2015densest}. Esfandiari et al. extend the analysis of uniform sampling of edges to several other problem. Mitzenmacher et al. study the efficiency of uniform edge sampling for densest subgraph in hypergraphs~\cite{mitzenmacher2015scalable}.

For the max cut problem, Kapralov, Khanna, and Sudan~\cite{kapralov2015streaming} and independently Kogan and Krauthgamer~\cite{kogan2015sketching}  showed
that a streaming $(1-\epsilon)$-approximation algorithm to estimate the size of max cut requires $n^{1-O(\epsilon)}$ space. Later, Kapralov, Khanna and Sudan~\cite{kapralov20171+} show that for some small $\epsilon$ any streaming $(1-\epsilon)$-approximation algorithm to estimate the size of max cut requires $\Omega(n)$ space. 
Very recently, Bhaskara et al.~\cite{bhaskara2018sublinear} provide a $2$-pass $(1-\epsilon)$-approximation streaming algorithm using $\tilde{O}(n^{1-\delta})$ space for graphs with average degree $n^{\delta}$.

Finally, when considering matching algorithms, there are numerous works on maximum matching in streaming and semi-streaming setting
~\cite{assadi2015tight,behnezhad2017brief,chitnis2015brief,chitnis2015parameterized,chitnis2016kernelization,konrad2012maximum,konrad2015maximum,esfandiari2015streaming,kapralov2014approximating}. Note that a maximal matching is a $0.5$ approximation to the maximum matching, and it is easy to provide one in the semi-streaming setting. However, improving this approximation factor in one pass is yet open. There are several works that improve this approximation factor in a few passes~\cite{ahn2011linear,behnezhad2017brief,konrad2012maximum,mcgregor2005finding}.
Maximum matching in hypergraphs has also been considered in the streaming setting~\cite{chitnis2016kernelization}.	
 
While a $0.5$-approximation for unweighted matching in the semi-streaming setting is trivial, such an approximation for weighted matching appears nontrivial. There is a sequence of works improving the approximation factor of weighted matching in the semi-streaming setting~\cite{bury2015sublinear,EpsteinLMS09,paz20172+}, and just recently Paz and Schwartzman provide a semi-streaming $(0.5-\epsilon)$-approximation algorithm.

There are, of course, many, many other related problems;  see~\cite{czumaj2010sublinear}, for example, for a survey on sublinear algorithms.

%


\section{Providing a Linear Sampling}\label{sec:sample}
In this section we provide a technique to construct the desired sampled graph $\H = (V,E_H)$ from a metric graph $G=(V,E_G)$.
We first provide a useful decomposition of the graph.  We show this decomposition allows us to obtain a graph
$\Ha$ that satisfies the first property of $\H$, namely that edge weights are scaled down (in expectation,
for edges with scaled weights less than 1) by a factor of $\alpha$.  
We then show how to determine a proper value $\alpha$ so that expected sum of the edge weight is between $\beta$
and $2\beta$ as desired. 

\subsection{A Graph Decomposition}

We start with a decomposition for a metric graph $G$,
assuming an upper bound $L$ on the weight of the edges. For an suitable number $t$ determined later we define the following sequences. 
\begin{itemize}
	\item A sequence of graphs $G = G_1\supseteq G_2\supseteq \dots \supseteq G_{t}$.
	\item A sequence of vertex sets $\nu_1,\dots,\nu_t$.
	\item A sequence of weights, $L_1=L, L_2=\frac{L_1}2, \dots, L_{t} = \frac{L_{t-1}}2$.
\end{itemize}
We denote the vertex set and edge set of $G_i$ by $V_i$ and $E_i$ respectively.
We begin with $G_1 = G$, and $G_i$ is constructed from $G_{i-1}$ by removing vertices in $\nu_{i-1}$, i.e. $G_i=G_{i-1}\setminus \nu_{i-1}$.  However, defining $\nu_i$, which depends on $G_i$, requires the following additional definitions.  
For any $i\in \{1,\dots,t\}$ and $\Lambda \in [0,1]$, define $G_i^{\Lambda} = (V_i,E_i^{\Lambda})$ to be the graph obtained by removing all edges with weight less than $\Lambda L_i$ from $G_i$, i.e., $e \in E_i^{\Lambda}$ if and only if $e \in E_i$ and $w_e \geq \Lambda L_i$. 

We now define $G_i$ and $\nu_i$ iteratively as follows. We define $\nu_{i,\xi}$ to be the set of vertices in $G_i^{\lambda/4}$ with degree at least $\xi|V_i|$. We let $\nu_i$ be an arbitrary subset such that $\nu_{i,1/2} \subseteq \nu_i \subseteq \nu_{i,1/4}$. As mentioned, $G_1=G$ and $G_i=G_{i-1}\setminus \nu_{i-1}$. We define $E_{\nu_i} = E_i \setminus E_{i+1}$. Note that $E_{\nu_i}$ is the set of edges neighboring $\nu_i$ in $G_i$. 

\begin{lemma}\label{lm:sketch:vc}
	For any $i\in \{1,\dots,t\}$, the set of vertices in $G_i^{\lambda/4}$ with degree at least $\frac{|V_i|}{2}$ (i.e., $\nu_{i,1/2}$) is a vertex cover for $G_i^{1/2}$.
\end{lemma}
\begin{proof}
	Let $(u,v)\in E_i^{1/2}$ be an edge in $G_i^{1/2}$, and let $d_u$ and $d_v$ be the degrees of $u$ and $v$ in $G_i^{\lambda/4}$ respectively. Next we show that $d_u+d_v \geq |V_i|$. Hence we have $d_u\geq \frac{|V_i|}{2}$ or $d_v \geq \frac{|V_i|}{2}$. This means that $\nu_{i,1/2}$ covers $(u,v)$ as desired.
	
	Notice that, $(u,v)\in E_i^{1/2}$ means that $w_{u,v} \geq \frac{L_i}{2}$. Hence, by the  $\lambda$-triangle inequality, for any $v' \in V_i$ we have $w_{(u,v')} + w_{(v',v)} \geq \lambda w_{(u,v)} \geq \frac{\lambda L_i}{2}$. Now we are ready to bound $d_u+d_v$.
	\begin{align*}
		d_u+d_v &=  \Big( 1 + \sum_{v' \in V_i\setminus \{u,v\}} 1_{(u,v') \in E_i^{\lambda/4}} \Big) + \Big( 1 + \sum_{v' \in V_i\setminus \{u,v\}} 1_{(v,v') \in E_i^{\lambda/4}} \Big) &\text{Extra $1$ is for $(u,v)$}\\
		&= 2 + \sum_{v' \in V_i\setminus \{u,v\}} \Big( 1_{(u,v') \in E_i^{\lambda/4}} + 1_{(v,v') \in E_i^{\lambda/4}} \Big)\\
		&= 2 + \sum_{v' \in V_i\setminus \{u,v\}} \Big( 1_{w_{(u,v')} \geq {\lambda L_i}/{4}} + 1_{w_{(v,v')} \geq {\lambda L_i}/{4}} \Big)\\
		& \geq 2 + \sum_{v' \in V_i\setminus \{u,v\}} 1 &\text{Since $w_{(u,v')}+ w_{(v',v)}  \geq \frac{\lambda L_i}{2}$}\\
		& = 2 + |V_i| - 2 = |V_i|,
	\end{align*}
	which completes the proof.
\end{proof}

\begin{lemma}\label{lm:sketch:L}
	For any $i\in \{1,\dots,t\}$, $L_i$ is an upper bound on weight of the edges in $G_i$, i.e., we have $\max_{e\in E_i} w_e \leq L_i$.
\end{lemma}
\begin{proof}
	For $i=1$, $L_1=L$ which is an upper bound on weight of the edges in $G_1=G$. For $i>1$,  $\nu_{i-1,1/2}$ is a vertex cover of $G_{i-1}^{1/2}$, by Lemma \ref{lm:sketch:vc}. Moreover, by definition we have $\nu_{i-1,1/2}\subseteq \nu_{i-1}$. Hence $\nu_{i-1}$ is a vertex cover of $G_{i-1}^{1/2}$. This means that every edge with weight at least $\frac{L_{i-1}}{2}$ has a neighbor in $\nu_{i-1}$. Recall that $G_i=G_{i-1}\setminus \nu_{i-1}$, and hence, $G_i$ has no edge with weight at least $\frac{L_{i-1}}{2} = L_i$.
\end{proof}

The following theorem compares the average weight of the edges in $E_{\nu_i}$ with $L_i$. We later use this in Theorem \ref{thm:sketch:Ha} to bound the number of queries.

\begin{lemma} \label{lm:sketch:costbound}
	For any $i\in \{1,\dots,t\}$, we have
	$$ \frac{\lambda}{32} L_i |V_i||\nu_i| \leq \sum_{e \in E_{\nu_i} } w_e \leq L_i |V_i||\nu_i|.$$
\end{lemma}
\begin{proof}
	We start by proving the upper bound. Recall that $E_{\nu_i}$ is the set of edges neighboring $\nu_i$ in $G_i$. Hence the number of edges in $E_{\nu_i}$ is upper bounded by sum of the degrees of the vertices of $\nu_i$ in $G_i$. The degree of each vertex in $G_i$ is $|V_i|-1 < |V_i|$, and there are $|\nu_i|$ vertices in $\nu_i$. Thus, we have $|E_{\nu_i}|\leq |V_i||\nu_i|$. Moreover, by Lemma \ref{lm:sketch:L}, for each $e \in E_{\nu_i} \subseteq E_i$ we have $w_e \leq L_i$. Therefore we have
	\begin{align*}
	\sum_{e \in E_{\nu_i} } w_e \leq \sum_{e \in E_{\nu_i} } L_i \leq L_i |V_i||\nu_i|.
	\end{align*} 
	
	Next we prove the lower bound. Recall that we have $\nu_i \subseteq \nu_{i,1/4}$. Thus, for each $v\in \nu_i$, the degree of $v$ in $G_i^{\lambda/4}$ is at least $\frac{|V_i|}{4}$. Thus, for any \emph{fixed} $v\in \nu_i$ we have
	\begin{align}\label{eq:sketch:lowbound}
	\sum_{(u,v) \in E_{\nu_i}} w_{(u,v)} &= \sum_{(u,v) \in E_i} w_{(u,v)} &\text{Definition of $E_{\nu_i}$ for $v \in \nu_i$}\\ 
	&\geq \sum_{(u,v) \in E_i^{\lambda/4}}  w_{(u,v)} &G_i^{\lambda/4}\subseteq G_i\nonumber\\
	&\geq \sum_{(u,v) \in E_i^{\lambda/4}} \frac{\lambda L_i}{4} & \text{Definition of $G_i^{\lambda/4}$}\nonumber\\
	&\geq \frac{|V_i|}{4} \frac{\lambda L_i}{4} = \frac {\lambda} {16} |V_i|L_i. &v \in \nu_i \subseteq \nu_{i,1/4}
	\end{align}
	Note that each edge in $E_{\nu_i}$ intersects at most two vertices in $\nu_i$. Therefore, we have
	\begin{align*}
	\sum_{e \in E_{\nu_i} } w_e &\geq \frac 1 2 \sum_{v \in \nu_i} \sum_{(u,v) \in E_{\nu_i}} w_{(u,v)} \\
	&\geq \frac 1 2 \sum_{v \in \nu_i} \frac {\lambda} {16} |V_i|L_i &\text{By Inequality \ref{eq:sketch:lowbound}}\\
	&\geq \frac{\lambda}{32} L_i |V_i||\nu_i|,
	\end{align*}
	which completes the proof of the lemma.
\end{proof}

Lemma~\ref{lm:sketch:makeSeq} provides a technique to construct $\nu_i$ using $\tilde{O}(n)$ queries, with high probability. This to prove this lemma we sample some edges. Notice that these sampled edges are different from the edges that we sample to keep in $\H$.
We use the following standard version of the Chernoff bound (see e.g. \cite{mitzenmacher2005probability}) in Lemma~\ref{lm:sketch:makeSeq} as well as the rest of the paper.
\begin{lemma}[Chernoff Bound]
	Let $x_1,x_2,\dots,x_r$ be a sequence of independent binary (i.e., $0$ or $1$) random variables, and let $X=\sum_{i=1}^{r}x_i$. For any $\epsilon \in [0,1]$, we have 
	\begin{align*}
	\Pr\left(|X-\E[X]| \geq \epsilon \E[X] \right) \leq 2 \exp(-\epsilon^2\E[X]/3).
	\end{align*}
\end{lemma}

As we are now moving to doing sampling, we briefly remark on some
noteworthy points.  First, there is some probability of failure in our results.  We
therefore refer to the success probability in our results, and note
that our algorithms may fail ``silently''; that is, we may not realize
the algorithm has failed (because of a low probability event in the
sampling).  Also, we emphasize that in general, in what follows, when
referring to the number of queries required, we mean the expected
number of queries. However, using expectations is for convenience; all
of our results throughout the paper could instead be turned into
results bounding the number of queries required with high probability
(say probability $1-O(1/n)$ using Chernoff bounds at the cost of at
most constant factors in the standard way.  Finally, in some places we
may sample which edges we decide to query from a set of edges with a
fixed probability $p$.  In such situations, instead of iterating
through each edge (which could take time quadratic in the number of
vertices) we can generate the number of samples from a binomial
distribution and then generate the samples without replacement;
alternatively, we could determine which sample is the next sample at
each step using by calculating a geometrically distributed random
variable.  We assume this work can be done in constant time per sample.
For this reason, our time depends on the number of queries, and not
the total number of edges.

\begin{lemma}\label{lm:sketch:makeSeq}
	For any $i\in \{1,\dots,t\}$, given $G_i$ and $L_i$, one can construct $\nu_i$ using $192 (\log n +\log t) |V_i| \in \tilde{O}(n)$ expected queries, succeeding with probability at least $1-\frac{1}{n t}$.
\end{lemma}
\begin{proof}
	If $|V_i| \leq 384 (\log n +\log t)$, we have $E_i = {|V_i| \choose 2} = \frac 1 2 384 (\log n +\log t) (|V_i|-1) \leq 192 (\log n +\log t) |V_i|$. Hence in this case we query all the edges and construct $\nu_i$. In what follows we assume $|V_i| \geq 384 (\log n +\log t)$. To construct $\nu_i$ we sample each edge in $E_i$ with probability $p = \frac {384 (\log n +\log t)}{|V_i|}$. We add a vertex $v$ to $\nu_i$ if and only if at least $\frac {3}{8} 384 (\log n +\log t)$ of its sampled neighbors has weight $\frac {\lambda L_i}{4}$.  The number of sampled edges is $p {|V_i| \choose 2} = \frac 1 2 384 (\log n +\log t) (|V_i|-1) \leq 192 (\log n +\log t) |V_i| $.
	
	We denote the degree of a vertex $v\in V_i$ in $G_i^{\lambda /4}$ by $d_v$. Let $Y_e$ be a binary random variable that is $1$ if we sample $e$ and $0$ otherwise. Let us define $Z_e = Y_e 1_{w_e \geq \lambda L_i/4}$ and $Z_v = \sum_{u\in V_i} Z_{(u,v)}$. Recall that we add $v \in V_i$ to $\nu_i$ if and only if $Z_v \geq \frac {3}{8} 384 (\log n +\log t)$. Notice that, for any $v\in V_i$ we have 
	\begin{align}\label{eq:sketch:makeSeq:E}
		\ex{Z_v} = \ex{\sum_{u\in V_i} Z_{(u,v)}}= \sum_{u\in V_i} \ex{Y_{(u,v)}} 1_{w_{u,v} \geq \lambda L_i/4} = p \sum_{u\in V_i} 1_{w_{u,v} \geq \lambda L_i/4} = p d_v.
	\end{align}
	As $Z_v$ is the sum of independent binary random variables, by the Chernoff bound we have
	\begin{align*}
		\prob{|Z_v - \ex{Z_v}| \geq \frac{384 (\log n + \log t)}{8}} &\leq 2 \exp \Big( - \frac 1 3 \big(\frac{384 (\log n + \log t)}{8 \ex{Z_v}}\big)^2 \ex{Z_v} \Big) &\text{Chernoff bound}\\
		&= 2 \exp \Big( -  \frac{768 (\log n + \log t)^2}{ \ex{Z_v}} \Big) \\
		&= 2 \exp \Big( -  \frac{768 (\log n + \log t)^2}{ p d_v} \Big) & \ex{Z_v} = p d_v\\
		&= 2 \exp \Big( -  \frac{2 (\log n + \log t) |V_i|}{  d_v} \Big) & p =  \frac {384 (\log n +\log t)}{|V_i|}\\
		&\leq 2 \exp \Big( -  2 (\log n + \log t) \Big) & d_v \leq |V_i|\\
		& = \frac{2}{n^2t^2} \leq \frac{1}{n^2t}. &\text{Assuming $t\geq 2$} 
	\end{align*}
	By applying the union bound we have
	\begin{align*}
		\prob{\exists_{v \in V_i} |Z_v - \ex{Z_v}| \geq {48 (\log n + \log t)}} \leq \sum_{v \in V_i} \prob{|Z_v - \ex{Z_v}| \geq {48 (\log n+ \log t)}} = |V_i| \frac{1}{n^2t} \leq \frac{1}{nt}.
	\end{align*}
	This means that with probability at least $1-\frac{1}{nt}$, simultaneously for all vertices $v\in V_i$ we have 
	\begin{align}\label{eq:sketch:beta2alpha:Zv}
		|Z_v - \ex{Z_v}| \leq {48 (\log n + \log t)}.
	\end{align}	
	Next assuming that for all vertices $v\in V_i$ we have $|Z_v - \ex{Z_v}| \leq {48 (\log n + \log t)}$ we show that the $\nu_i$ that we pick satisfy the property $\nu_{i,1/2} \subseteq \nu_i \subseteq \nu_{i,1/4}$. 
	
	Applying Equality \ref{eq:sketch:makeSeq:E} to Inequality \ref{eq:sketch:beta2alpha:Zv} gives us $|Z_v - pd_v| \geq {48 (\log n + \log t)}$. By replacing $p$ with $\frac {384 (\log n +\log t)}{|V_i|}$ and rearranging the inequality we have 
	\begin{align*}
		Z_v  \leq \frac {384 (\log n +\log t)}{|V_i|} d_v + {48 (\log n + \log t)} & < \frac 3 8 {384 (\log n +\log t)} &\text{assuming $d_v < \frac{|V_i|}{4}$}.
	\end{align*}
	This means that if $d_v < \frac{|V_i|}{4}$ we have $Z_v < \frac 3 8 384 (\log n +\log t)$ and hence $v \notin \nu_i$. Therefore we have $ \nu_i \subseteq \nu_{i,1/4}$.
	Similarly, we have
	\begin{align*}
		Z_v  \geq \frac {384 (\log n +\log t)}{|V_i|} d_v - {48 (\log n + \log t)} & \geq \frac 3 8 384 (\log n +\log t) &\text{assuming $d_v \geq \frac{|V_i|}{2}$}.
	\end{align*}
	This means that if $d_v \geq \frac{|V_i|}{2}$ we have $Z_v \geq \frac 3 8 384 (\log n +\log t)$ and hence $d_v \in \nu_i$. Therefore we have $\nu_{i,1/2} \subseteq \nu_i$.	
	
\end{proof}

Finally, for completeness we use the following lemma to find a good upper bound $L$ on $\max_{e\in E} w_e$ in order to start our construction of the graph decomposition (which required an upper bound on the weight of the edges).  
\begin{lemma}\label{lm:sketch:L2}
	For any $\lambda$-metric graph $G=(V,E)$, one can compute a number $L$ such that $\max_{e\in E} w_e \leq L\leq \frac {2}{\lambda} \max_{e\in E} w_e$ using $n-1$ queries.
\end{lemma}
\begin{proof}
	Let $v'\in V$ be an arbitrary vertex. We set $L = \frac {2}{\lambda} \max_{u'\in V} w_{u',v'}$. Note that, one can simply query all the $n-1$ neighbors of $v'$ and calculate $L$. Clearly, we have $L = \frac {2}{\lambda} \max_{u'\in V} w_{u',v'} \leq \frac {2}{\lambda} \max_{e\in E} w_e$. Next, we show that $\max_{e\in E} w_e \leq L$.
	
	Let $(u,v)$ be an edge such that $w_{(u,v)}=\max_{e\in E} w_e$. If $v' \in \{u,v\}$ we have $\max_{u'\in V} w_{u',v'} = \max_{e\in E} w_e $ which directly implies $L\leq \frac {2}{\lambda} \max_{e\in E} w_e$ as desired. Otherwise, note that by the $\lambda$-triangle inequality we have $w_{(u,v')}+w_{(v,v')}\geq \lambda w_{(u,v)}$. Thus, we have $ \max(w_{(u,v')} , w_{(v,v')}) \geq \frac{\lambda} {2} w_{(u,v)} $. Therefore, we have
	\begin{align*}
		L = \frac {2}{\lambda} \max_{u'\in V} w_{u',v'} \geq \frac {2}{\lambda} \max(w_{(u,v')} , w_{(v,v')}) \geq w_{(u,v)} = \max_{e\in E} w_e,
	\end{align*}
	as desired.
\end{proof}

	\subsection{Constricting $\Ha$}\label{sec:Ha}
	We know show how to construct what we call $\Ha$, which is derived from our original metric graph $G$.  Recall
$\Ha$ has the property that for each original edge $e$ of weight $w_e$, independently, if $\alpha w_e > 1$, then $\Ha$
contains edge $e$ with weight $\alpha w_e$, and if $\alpha w_e < 1$, then $\Ha$ contains edge $e$ with weight 1 with
probability $\alpha w_e$.  

We define $\overline{w}$ to be the average of the weight of edges in $G$. We use this notion in the following lemma as well as Lemma \ref{lm:sketch:beta2alpha} and Theorem \ref{thm:sketch:mainH}.

The following theorem constructs $\Ha$ using an expected $O(n \log^2 n + n \log^2 \max_{e\in E} \alpha w_e  + \alpha \overline{w} {n \choose 2})$ queries. 

\begin{theorem}\label{thm:sketch:Ha}
	For any $\alpha$ one can construct $\Ha$ using $O\big(n \log^2 n+ n \log^2 \max_{e\in E} \alpha w_e + \frac{1}{\lambda} \alpha \overline{w} {n \choose 2}+\frac n {\lambda}\big)$ queries in expectation, succeeding with probability at least $1-\frac{1}{n}$. 
\end{theorem}
\begin{proof}
	By Lemma \ref{lm:sketch:L2} we find an upper bound $L$ on the weight of the edges, using $n-1$ queries.
	Recall that $t$ is the number of graphs in our decomposition.  We set $t= \log_2 n + \log_2 \max_{e\in E} \alpha w_e$. Given $L_1=L$, by definition we have 
	\begin{align}\label{eq:sketch:Lt}
	L_t = \frac{L}{2^t} &= \frac{L}{n \max_{e\in E} \alpha w_e} &  t= \log_2 n + \log_2 \max_{e\in E} \alpha w_e \nonumber\\
	&\leq  \frac{2}{\lambda \alpha n}. &L\leq \frac{2}{\lambda}\max_{e\in E} w_e
	\end{align}

	Recall that, using Lemma \ref{lm:sketch:makeSeq}, one can construct $\nu_i$ and thus $V_{i+1}$ using $192 (\log n +\log t) n $ queries, succeeding with probability at least $1-\frac{1}{ n t}$. We start with $G_1=G$ and iteratively apply Lemma \ref{lm:sketch:makeSeq} to construct the sequence $G_1,\dots,G_t$ and $\nu_1,\dots,\nu_t$. We apply Lemma \ref{lm:sketch:makeSeq} $t$ times, and hence using a union bound, all of the $G_i$ were successfully constructed with probability at least $1-t\frac 1 {nt} = 1-\frac 1 {n}$. Next, we show how to construct $\Ha$ assuming the sequences $G_1,\dots,G_t$ and $\nu_1,\dots,\nu_t$ are valid. Note that constructing the graph decomposition we use at most $192  (\log n +\log t) n \times t \in O(n \log^2 n + n \log^2 \max_{e\in E} \alpha w_e ) $ queries.
	
	Recall that $E_{\nu_i} = E_i \setminus E_{i+1}$. Also, we have $E_1=E$. Thus, the sequence $E_{\nu_1},\dots,E_{\nu_{t-1}}$ is a decomposition of $E\setminus E_t$.  Also, note that $E_{\nu_i} = \{(u,v)\big| u\in\nu_i \text{ and } v\in V_i\}$. Therefore, given $G_1,\dots,G_t$ and $\nu_1,\dots,\nu_t$ we can decompose the edge set $E$ into $E_{\nu_1},\dots,E_{\nu_t}$.
	
	Let $j$ be the smallest index such that $L_j \leq \frac{1}{\alpha}$. Notice that $j\leq t$ by Inequality \ref{eq:sketch:Lt}. For each $i\in \{1,\dots,j-1\}$ we query each edge $e \in E_{\nu_i}$. 
	If $\alpha w_e > 1$, add edge $e$ with weight $\alpha w_e$ to $E_H$. If $\alpha w_e\leq 1$, we add edge $e$ with weight $1$ to $E_H$ with probability $\alpha w_e$ independently.

	For each $i\in \{j,\dots,t-1\}$ we query each edge $e \in E_{\nu_i}$ with probability $\alpha L_i$. We add a queried edge $e$ to $E_H$ with probability $\frac {w_e}{L_i}$ and withdraw it otherwise. Note that $\alpha L_i \leq \alpha L_j \leq \alpha \frac {1}{\alpha}  = 1$. Also $L_i$ is an upper bound on the edge weights in $E_i \supseteq E_{\nu_i}$, and thus $\frac {w_e}{L_i}\leq 1$. Therefore, the probabilities $\alpha L_i$ and $\frac {w_e}{L_i}$ are valid. Also, notice that we add each edge to $E_H$ with probability $\alpha L_i \times \frac {w_e}{L_i} = \alpha w_e$ as desired. 
	
	For each edge $e\in E_t$ we query $e$ with probability $\frac 2 {\lambda n}$. We add a queried edge $e$ to $E_H$ with probability $\frac { \lambda  n}{2}\alpha w_e$ and withdraw it otherwise. Recall $L_t$ is an upper bound on the edges edges weights in $E_t$, and by Inequality \ref{eq:sketch:Lt} we have $L_t \leq  \frac{2}{\lambda\alpha n}$. Thus, we have
	\begin{align*}
		\frac { \lambda  n}{2}\alpha w_e \leq \frac { \lambda  n}{2}\alpha L_t 
		\leq \frac { \lambda  n}{2}\alpha \frac{2}{\lambda \alpha n}  = 1. 
	\end{align*}
	Therefore, $\frac { \lambda  n}{2}\alpha w_e$ is a valid probability. Again, notice that we add each edge to $E_H$ with probability $\frac 2 {\lambda n} \times \frac { \lambda  n}{2}\alpha w_e= \alpha w_e$ as desired. 
	Next we bound the total number of edges that we query.

	Let $Y_e$ be a random variable that is $1$ if we query $e$ and $0$ otherwise. We bound the expected number of edges that we query by
	\begin{align*}
		\ex{\sum_{e\in E} Y_e} &= \sum_{e\in E} \ex{Y_e} \\
		&= \sum_{i=1}^{j-1} \sum_{e \in E_{\nu_i} }  \ex{Y_e}
		+ \sum_{i=j}^{t-1} \sum_{e \in E_{\nu_i} }  \ex{Y_e}
		+ \sum_{e \in E_{t} }  \ex{Y_e}
		 &\text{$E_{\nu_1},\dots,E_{\nu_{t-1}},E_t$ is a decomposition of $E$}\\
		&= \sum_{i=1}^{j-1} \sum_{e \in E_{\nu_i} }  1
		+ \sum_{i=j}^{t-1} \sum_{e \in E_{\nu_i} }  \alpha L_i
		+ \sum_{e \in E_{t} }  \frac 2 {\lambda n}\\
		&\leq \sum_{i=1}^{t} \sum_{e \in E_{\nu_i} }  \alpha L_i + \sum_{e \in E_{t} }  \frac 2 {\lambda n} & \forall_{i < j} \alpha L_i \geq 1  \\
		&\leq \sum_{i=1}^{t} \sum_{e \in E_{\nu_i} }  \alpha L_i + \frac{n}{\lambda} & |E_t|\leq {n \choose 2}  \\
		&\leq \alpha \sum_{i=1}^{t}  |\nu_i| |V_i| L_i + \frac{n}{\lambda}  &  |E_{\nu_i}| \leq |\nu_i| |V_i| \\
		&\leq \alpha \sum_{i=1}^{t} \sum_{e \in E_{\nu_i} } \frac{32}{\lambda} w_e + \frac{n}{\lambda} &  \text{By Lemma \ref{lm:sketch:costbound}} \\
		&\leq \frac{32}{\lambda} \alpha \sum_{e \in E } w_e + \frac{n}{\lambda} &  \text{$E_{\nu_1},\dots,E_{\nu_t}\subseteq E$ are disjoint} \\
		&= \frac{32}{\lambda} \alpha {n \choose 2} \overline{w} + \frac{n}{\lambda}\\
		& \in O(\frac 1 {\lambda} \alpha {n \choose 2} \overline{w}+ \frac n {\lambda}).
	\end{align*}
	
	We used $O(n \log^2 n + n \log^2 \max_{e\in E} \alpha w_e )$ queries to construct the sequences $G_1,\dots,G_t$ and $\nu_1,\dots,\nu_t$, and used $O(\frac 1 {\lambda} \alpha {n \choose 2} \overline{w}+ \frac n {\lambda})$ queries to construct $\Ha$ based on these sequences. Therefore, in total we used $O\big(n \log^2 n+ n \log^2 \max_{e\in E} \alpha w_e + \frac{1}{\lambda} \alpha \overline{w} {n \choose 2}+\frac n {\lambda}\big)$ queries in expectation.
\end{proof}

	\subsection{Constructing $\H$}\label{sec:Hb}
	The following lemma relates $\i$ with $\alpha$. We use this to construct $\H$ using $\Ha$.

\begin{lemma}\label{lm:sketch:beta2alpha}
	 Let $\gamma\in [1,\infty)$ be an arbitrary number. Let $\frac{1}{\gamma} \overline{w} \leq \hat{w} \leq \overline{w}$, $\alpha = \frac {\i}{  {n \choose 2}\hat w}$, and $\Ha = (V,E_H)$. We have 
	$$\i \leq \ex{\sum_{e\in E_H} w'_e} \leq  {\gamma} \i, $$ where $w'_e$ is the weight of $e$ in $\Ha$.
\end{lemma}
\begin{proof}
	We have 
	\begin{align*}
	\ex{\sum_{e\in E_H} w'_e} &= \sum_{e\in E} \ex{w'_e} =  \sum_{e\in E} \alpha w_e  \\
	&=   \sum_{e\in E} \frac {\i}{  {n \choose 2}\hat w} w_e = \frac{\i}{ \hat w }\frac {\sum_{e\in E} w_e}{{n \choose 2}} &\text{Definition of $\alpha$}\\
	&= \frac{ \overline{w}}{\hat w } \i. &\text{Definition of $\overline{w}$}
	\end{align*}
	This together with $\frac 1 {\gamma} \overline{w} \leq \hat{w}$ gives us 
	\begin{align*}
	\ex{\sum_{e\in E_H} w'_e} =  \frac{\overline{w}}{ \hat w} \i \leq  \gamma \i.
	\end{align*}
	Similarly, by applying $\hat{w} \leq \overline{w}$ we have 
	\begin{align*}
	\ex{\sum_{e\in E_H} w'_e} =  \frac{\overline{w}}{ \hat w} \i \geq  \i.
	\end{align*}

\end{proof}

Lemma \ref{lm:sketch:w} shows how to estimate $\overline{w}$. We use this lemma together with Lemma \ref{lm:sketch:beta2alpha} to find a proper $\alpha$ based on the desired $\beta$ to construct $\H$. We note that in a metric space, i.e. $\lambda=1$, the following lemma gives a $1-\epsilon$ approximation of the average weight of the edges using $\tilde{O}(\frac n {\epsilon^2})$ queries, while the previous algorithm of Indyk~\cite{indyk1999sublinear} uses $O(\frac n {\epsilon^{3.5}})$ queries\footnote{Note that the algorithm in \cite{indyk1999sublinear} works with a constant probability while our algorithm works with probability $1-\frac 1 n$. The previous algorithm requires an extra logarithmic factor to work with probability $1-\frac 1 n$.}. In the next section, using $\H$ we improve this lemma and estimate the average weight of the edges using only $\tilde{O}(n+\frac{1}{\epsilon^2})$ queries.

\begin{lemma}\label{lm:sketch:w}
	For $\epsilon \in (0,1]$, one can find an estimator $\hat{w}$ of the average weight of the edges $\overline{w}$ such that $(1-\epsilon) \overline{w} \leq \hat{w} \leq (1+\epsilon) \overline{w}$, with probability $1-\frac 2 n$, using $O \big(  n \log^2 n +  \frac{n \log n }{\epsilon^2\lambda}\big) \in \tilde{O}(\frac n {\epsilon^2 \lambda})$ queries. 
\end{lemma}
\begin{proof}
	We first use $O(\frac{n}{\lambda})$ queries to provide an estimate $\hat{w}'$ such that $\frac{1}{2n} \overline{w} \leq \hat{w}' \leq \overline{w}$. Next we set $\alpha = \frac {\i}{  {n \choose 2}\hat w}$ and construct a corresponding $\Ha$. We use Lemma \ref{lm:sketch:beta2alpha} and Theorem \ref{thm:sketch:Ha} to lower bound the total weight of sampled edges by $\frac{ 3\log (2n) }{\epsilon^2}$ and upper bound the number of queries by $O \big(  n \log^2 n +  \frac{n \log n }{\epsilon^2\lambda}\big)$. At the end we use the lower bound on the total weight of sampled edges to show that the average weight of edges in $\Ha$ is concentrated around $\overline{w}$. 
	
	Let $v$ be an arbitrary vertex. We have 
	\begin{align*}
	\sum_{u \in V\setminus \{v\}} w_{u,v} &=
	\frac 1 n \sum_{u \in V\setminus \{v\}} w_{u,v} 
	+ \frac {n-1} n  \sum_{u \in V\setminus \{v\}}  w_{u,v} \\
	&\geq 		
	\frac 1 n \sum_{u \in V\setminus \{v\}} w_{u,v} 
	+ \frac {n-1} n  \frac{\lambda}{n-2} \sum_{u,u' \in V\setminus \{v\}}  w_{u,u'} &\text{$\lambda$-triangle inequality}\\
	&> 
	\frac{\lambda}{n} \sum_{e \in E}  w_{e}\\
	\end{align*}
	Hence for a set $S\subseteq V$ with $|S|=\lceil \frac 1 {\lambda} \rceil$ we have
	\begin{align*}
	\sum_{v\in S} \sum_{u \in V\setminus \{v\}} w_{u,v} > \sum_{v\in S} \frac{\lambda}{n} \sum_{e \in E}  w_{e} \geq \frac 1 n \sum_{e \in E}  w_{e}.
	\end{align*}
	On the other hand every edge appears at most twice in $\sum_{v\in S} \sum_{u \in V\setminus \{v\}} w_{u,v}$ and hence we have $\sum_{v\in S} \sum_{u \in V\setminus \{v\}} w_{u,v} \leq 2 \sum_{e \in E}  w_{e}$. Therefore, by setting $\hat{w}'= \frac{1}{2{n \choose 2}}\sum_{v\in S} \sum_{u \in V\setminus \{v\}} w_{u,v}$ we have $\frac 1 {2n} \overline{w} \leq \hat{w}' \leq \overline{w}$. Hence, one can query at most $\frac{n}{\lceil \lambda \rceil}$ edges to find a number $\hat{w}'$ such that $\frac 1 {2n} \overline{w} \leq \hat{w}' \leq \overline{w}$.
	
	Next, we set $\alpha=\frac{ 3\log (2n) }{\epsilon^2 {n \choose 2}\hat{w'}}$. By Lemma \ref{lm:sketch:beta2alpha} we have 
	\begin{align}\label{eq:sketch:Wedges}
		\frac{ 3\log (2n) }{\epsilon^2} \leq \ex{\sum_{e\in E_H} w'_e} \leq  \frac{6n \log (2n) }{\epsilon^2}, 
	\end{align}	
	where $w'_e$ is the weight of $e$ in $\Ha$. By Lemma \ref{thm:sketch:Ha}, with probability $1-\frac 1 n$, the expected number of queries we need to construct $\Ha$ is at most
	\begin{align*}
		O \Big( n \log^2 n+ n \log^2 \max_{e\in E} \alpha w_e + \frac{1}{\lambda} \alpha \overline{w} {n \choose 2}+\frac n {\lambda} &\Big) \in \\ 
		O \Big(  n \log^2 n+ n \log^2 \Big(\frac{3 n \log (2n) }{\epsilon^2}\Big) + \frac{1}{\lambda} \frac{6n \log(2 n) }{\epsilon^2}+\frac n {\lambda}&\Big) \in \\
		O \Big(  n \log^2 n +  \frac{n \log n }{\epsilon^2\lambda}&\Big). &\text{Assuming $\epsilon \geq \frac 1 {n}$ w.l.o.g.}
	\end{align*} 
	
	Now we set $\hat{w} = \frac 1 {{n \choose 2}} \sum_{e\in E_H} \frac{w'_e}{\alpha}$, where $w'_e$ is the weight of $e$ in $\Ha$. To complete the proof we show that $(1-\epsilon) \overline{w} \leq \hat{w} \leq (1+\epsilon) \overline{w}$, with probability $1-\frac 1 n$. Notice that 
	\begin{align}\label{eq:sketch:EWhat}
		\ex{\hat{w}} = \ex{\frac 1 {{n \choose 2}} \sum_{e\in E_H} \frac{w'_e}{\alpha}} = {\frac 1 {{n \choose 2}} \sum_{e\in E} w_e} = \overline{w}.
	\end{align} 
	Let $\chi_e$ be a binary random variable that indicates whether $\chi_e$ is sampled in $\Ha$ or not. Note that
	\begin{align}
		 \hat{w} - \ex{\hat{w}} 
		&= \frac 1 {{n \choose 2}} \sum_{e\in E_H} \frac{w'_e}{\alpha} - \frac 1 {{n \choose 2}} \sum_{e\in E} \ex{\frac{w'_e}{\alpha}}\nonumber\\
		& = \frac 1 {{n \choose 2}{\alpha}} \Big( \sum_{e\in E_H} {w'_e} - \sum_{e\in E} \ex{w'_e} \Big)\nonumber\\
		& = \frac 1 {{n \choose 2}{\alpha}} \Big( \sum_{ w_e \leq \frac 1 {\alpha} } \chi_e - \sum_{w_e\leq \frac 1 {\alpha}} \ex{w'_e} \Big). &w'_e = \ex{w'_e}\text{ when } w_e>\frac 1 {\alpha}  \label{eq:sketch:^w-Ew}
	\end{align}
	Therefore, we have
	\begin{align*}
		\prob{|\hat{w} - \overline{w}| \leq \epsilon \overline{w} } &= \prob{|\hat{w} - \ex{\hat{w}}| \leq \epsilon \overline{w}} &\text{By Equality \ref{eq:sketch:EWhat}}\\
		&= \prob{\Big|\frac 1 {{n \choose 2}{\alpha}} \Big( \sum_{ w_e \leq \frac 1 {\alpha} } \chi_e - \sum_{w_e\leq \frac 1 {\alpha}} \ex{w'_e} \Big)\Big| \leq \epsilon \overline{w} } &\text{By Equality \ref{eq:sketch:^w-Ew}} \\	
		&= \prob{\Big| \sum_{ w_e \leq \frac 1 {\alpha} } \chi_e - \sum_{w_e\leq \frac 1 {\alpha}} \ex{w'_e} \Big| \leq \epsilon \alpha {n \choose 2} \overline{w} } \\
		&\leq 2\exp{\Big(-\frac 1 3  \big(\frac{\epsilon \alpha {n \choose 2} \overline{w}}{\sum_{w_e\leq \frac 1 {\alpha}} \ex{w'_e}}\big)^2 \sum_{w_e\leq \frac 1 {\alpha}} \ex{w'_e}\Big)} &\text{Chernoff Bound}\\
		&\leq 2\exp{\Big(-\frac 1 3  \frac{\epsilon^2 \alpha^2 {n \choose 2}^2 \overline{w}^2}{\sum_{w_e\leq \frac 1 {\alpha}} \ex{w'_e}} \Big)}\\			
		&= 2\exp{\Big(-\frac 1 3  \frac{\epsilon^2 \big(\sum_{e \in E} \ex{w'_e} \big)^2}{\sum_{w_e\leq \frac 1 {\alpha}} \ex{w'_e}} \Big)}\\			
		&\leq 2\exp{\Big(-\frac 1 3  {\epsilon^2 \sum_{e \in E} \ex{w'_e}} \Big)}\\			
		&\leq 2\exp{\Big(-\frac 1 3  {\epsilon^2 \frac{3\log (2n) }{\epsilon^2}} \Big)} & \text{By Inequality \ref{eq:sketch:Wedges}}\\		
		&= 2\exp{\big(-  { {\log (2 n) }} \big)}	= \frac 1 n.
	\end{align*}
	This means that with probability $1-\frac 1 n$ we have $(1-\epsilon) \overline{w} \leq \hat{w} \leq (1+\epsilon) \overline{w}$ as desired.
\end{proof}

The following theorem constructs $\H$ using $\tilde{O}(\frac{n+\i}{\lambda})$ queries, with high probability.

\begin{theorem}\label{thm:sketch:mainH}
	For any $\i$ one can construct $\H$ using expected $O(n\log^2 n + n\log^2 \i+\frac{\i}{\lambda}+\frac{n\log n}{\lambda})\in \tilde{O}(\frac{n+\i}{\lambda})$ expected queries, with probability of success at least $1-\frac 3 n$.
\end{theorem}
\begin{proof}
	First, using Lemma \ref{lm:sketch:w} we find an estimator $\hat{w}$ of the average weight of the edges $\overline{w}$ such that $\frac 1 2 \overline{w} \leq \hat{w} \leq \overline{w}$, with probability $1-\frac 2 {n}$, using $O(n \log^2 n+\frac{n \log n}{\lambda})$ expected queries. Lemma \ref{lm:sketch:beta2alpha} says that by picking $\alpha = \frac {\i}{ {n \choose 2}\hat w}$, we have $\i \leq \ex{\sum_{e\in E_H} w'_e} \leq 2\i $, where $w'_e$ is the weight of $e$ in $\Ha=(V,E_H)$.
	\footnote{Note that, for any $\eta \in (0,1]$, one can use lemma \ref{lm:sketch:w} to find $\hat{w}$ such that $\frac 1 {1+\eta} \overline{w} \leq \hat{w} \leq \overline{w}$, with probability $1-\frac 2 {n}$, using $O(n \log^2 n+\frac{n  \log n}{\eta^2\lambda})$ expected queries, and then apply Lemma \ref{lm:sketch:beta2alpha} to show that by picking $\alpha = \frac {\i}{ {n \choose 2}\hat w}$, we have $\i \leq \ex{\sum_{e\in E_H} w'_e} \leq (1+\eta)\i $. We use $\eta = 1$ throughout for convenience.} 
	By Theorem \ref{thm:sketch:Ha} one can construct $\Ha$ using $O(n \log^2 n+ n \log^2 \max_{e\in E} \alpha w_e  + \frac 1 {\lambda}\alpha \overline{w} {n \choose 2}+\frac n {\lambda})$ expected queries, with probability $1-\frac {1}{n}$. Note that, we have
	\begin{align*}
	 	\max_{e\in E} \alpha w_e &= \alpha \max_{e\in E} w_e \\
	 	&\leq \alpha {{n \choose 2} \overline{w}} & \overline{w} \geq \frac{ \max_{e\in E} w_e}{{n \choose 2}}\\
	   &= \frac {4\i}{ 3 {n \choose 2}\hat w} {{n \choose 2} \overline{w}} & \alpha = \frac {4\i}{ 3 {n \choose 2}\hat w}\\
	   & \leq  {2\i}   &\frac 2 3 \overline{w} \leq \hat{w}\\
	\end{align*}
	Also, we have
	\begin{align*}
	\alpha \overline{w} n^2 &=  \frac {4\i}{ 3 {n \choose 2}\hat w} \overline{w} {n \choose 2} &\text{By $\alpha=\frac {4\i}{ 3 {n \choose 2}\hat w} $}\\
	&=\frac {4\overline{w}}{ 3 \hat w} \i \\
	&\leq 2 \i. &\text{By $\frac 2 3 \overline{w} \leq \hat{w}$}
	\end{align*}
	By $\alpha \overline{w} n^2 \leq 2\i$ and $\max_{e\in E} \alpha w_e\leq 2\i$ we have
	\begin{align*}
	n \log^2 n+ n \log^2 \max_{e\in E} \alpha w_e  + \frac 1 {\lambda}\alpha \overline{w} {n \choose 2} + \frac n {\lambda}
	&\leq n \log^2 n + n log^2 ({2\i}) + \frac{2\i}{\lambda} + \frac n {\lambda} \\
	&\in O(n\log^2 n + n\log^2 \i + \frac{\i+n}{\lambda}).
	\end{align*}

	Therefore, the total number of expected queries is ${O}(\frac{n \log n}{\lambda}+n\log^2 n + n\log^2 \i+\frac{\i+n}{\lambda}) \in \tilde{O}(\frac{\i+n}{\lambda})$. We properly estimate $\hat{w}$ with probability at least $1-\frac{2}{n}$ and Theorem \ref{thm:sketch:Ha} holds with probability at least $1-\frac{1}{n}$. Therefore, by the union bound, the statement of this theorem holds with probability at least $1-  \frac{3}{n}$.
\end{proof}

\section{Applications of Linear Sampling}\label{sec:app}
In this section we use the sketch $\H$ to develop approximation algorithms for densest subgraph, maximum $k$-hypermatching, and maximum cut, as well as estimating the average distance.
We first define the problems and provide relevant notation. 
The densest subgraph of a graph $G=(V,E)$ is an induced subgraph of $G$, indicated by its set of vertices $S^*\subseteq V$, that maximizes $\frac{\sum_{ u,v \in S^* } w_{u,v} }{|S^*|}$. We indicate the value of the densest subgraph by $\optd$.
The max cut of a graph $G=(V,E)$ is a decomposition of the vertex set of $G$ into two sets $S^*, V\setminus S^* \subseteq V$, that maximizes $\sum_{ u \in S^*, v\in V\setminus S^* }  w_{u,v}$. We indicate the value of the max cut by $\optc$.
A $k$-hypermatching of a set of points $V$ is a decomposition of $V$ into a collection of $n/k$ sets $\S^*= \{S_1^*,S_2^*,\dots,S_{n/k}^*\}$, each of size $k$. One can also see this as covering a graph $G=(V,E)$ with clusters of size $k$. A maximum $k$-hypermatching is a $k$-hypermatching that maximizes $\sum_{i=1}^{n/k}\sum_{ u,v \in S_i^* }  w_{u,v}$. We use $\optm$ to indicate the value of the maximum $k$-hypermatching. 

For a sketch $\H=(V,E_H)$ we define random variables $X_{u,v}$ and $Y_{u,v}$. $Y_{u,v}$ is $0$ if $(u,v) \notin E_H$, and is equal to the weight of the edge $(u,v)$ in $\H$ otherwise. $X_{u,v}=1$ if $Y_{u,v}=1$ and $X_{u,v}=0$ otherwise. Recall that if we sample an edge $e$ with $\alpha w_e \leq 1$, weight of $e$ in $\H$ is $1$. Note that $\E[Y_{u,v}]= \alpha w_{u,v}$.

We first start with a simple application, using $\H$ to estimate the average weight of the edges using $\i=O(\frac{\log n }{\eps^2})$. This together with Theorem \ref{thm:sketch:mainH} allows us to find the average weight of the edges in a $\lambda$-metric space with probability $1-\frac{4}{n}$ using $O(n\log^2 n + \frac{\log n}{\lambda \eps^2}+\frac{n\log n}{\lambda})\in \tilde{O}(\frac{n+  1 /{\eps^2}}{\lambda})$ expected queries.\footnote{Again, we emphasize that we can turn these results into bounds with a corresponding upper bound on the queries, with a small increase in the failure probability.}   In particular for a metric space this gives a $1-\eps$ approximation of the average weight of the edges using $\tilde{O}(n + \frac 1 {\epsilon^2})$ queries.

In what follows (throughout this section), when considering the failure probability of the approximation algorithms, we assume that $\H$ has been constructed successfully.  That is, we provide for a failure probability in this stage of at most $1/n$, which when combined with Theorem \ref{thm:sketch:mainH} allows for our success probability of at least $1-\frac{4}{n}$ overall.

\begin{theorem}
	Take $\i = \frac{ 3\log (2n) }{\epsilon^2}$. We have
	$$ (1-\eps)\overline{w} \leq \frac {1}{\alpha {n \choose 2}}\sum_{e \in E} Y_{e} \leq (1+\eps) \overline{w},$$
	with probability at least $1-\frac 1 n$.
\end{theorem}

\begin{proof}
	We define $\hat w = \frac {1}{\alpha {n \choose 2}}\sum_{e \in E} Y_e$
	Notice that 
	\begin{align}\label{eq:app:EWhat}
	\ex{\hat{w}} = \ex{\frac 1 {{n \choose 2}} \sum_{e \in E} \frac{Y_e}{\alpha}} = {\frac 1 {{n \choose 2}} \sum_{e\in E} w_e} = \overline{w}.
	\end{align} 
	We have
	\begin{align*}
	\prob{|\hat{w} - \overline{w}| \leq \epsilon \overline{w} } &= \prob{|\hat{w} - \ex{\hat{w}}| \leq \epsilon \overline{w}} &\text{By Equality \ref{eq:app:EWhat}}\\
	&= \prob{\Big|\frac 1 {{n \choose 2}{\alpha}} \Big( \sum_{ e \in E } Y_e - \sum_{e \in E} \ex{Y_e} \Big)\Big| \leq \epsilon \overline{w} }\\
	&= \prob{\Big|\frac 1 {{n \choose 2}{\alpha}} \Big( \sum_{ e \in E } X_e - \sum_{e \in E} \ex{X_e} \Big)\Big| \leq \epsilon \overline{w} }&\text{If $Y_{e}\neq X_e$, $Y_{e}=\E[Y_{e}]$}\\
	&= \prob{\Big| \sum_{ e \in E } X_e - \sum_{e \in E} \ex{X_e} \Big| \leq \epsilon \alpha {n \choose 2} \overline{w} }\\
	&\leq 2\exp{\Big(-\frac 1 3  \big(\frac{\epsilon \alpha {n \choose 2} \overline{w}}{\sum_{e \in E} \ex{X_e}}\big)^2 \sum_{e \in E} \ex{X_e}\Big)} &\text{Chernoff Bound}\\
	&\leq 2\exp{\Big(-\frac 1 3  \frac{\epsilon^2 \alpha^2 {n \choose 2}^2 \overline{w}^2}{\sum_{e \in E} \ex{X_e}} \Big)}\\			
	&= 2\exp{\Big(-\frac 1 3  \frac{\epsilon^2 \big(\sum_{e \in E} \ex{Y_e} \big)^2}{\sum_{e \in E} \ex{X_e}} \Big)}\\			
	&\leq 2\exp{\Big(-\frac 1 3  {\epsilon^2 \sum_{e \in E} \ex{Y_e}} \Big)}\\			
	&\leq 2\exp{\Big(-\frac 1 3  {\epsilon^2 \frac{3\log (2n) }{\epsilon^2}} \Big)} & \i = \frac{ 3\log (2n) }{\epsilon^2} \\		
	&= 2\exp{\big(-  { {\log (2 n) }} \big)}	= \frac 1 n.
	\end{align*}
	This means that with probability $1-\frac 1 n$ we have $(1-\epsilon) \overline{w} \leq \hat{w} \leq (1+\epsilon) \overline{w}$ as desired.
\end{proof}

Next we provide our results for the densest subgraph problem.

\begin{theorem}
	Take $\i = \frac{9\log n}{\eps^2} n$. Let $S$ be a $\phi$-approximation solution to the densest subgraph problem on $\H$. $S$ is a $\phi-2\eps$ approximation solution to the densest subgraph on $G$, with probability at least $1-\frac 1 n$.
\end{theorem}
\begin{proof}
	 We start by lower bounding $\optd$.
	 \begin{align}\label{eq:dens:optmin}
		 	\optd \geq \frac{\sum_{u,v\in V} w_{u,v}}{|V|} = \frac{ \frac{1}{\alpha}\sum_{u,v\in V} \E[Y_{u,v}]}{n} \geq \frac{1}{\alpha} \frac{\beta}{n } = \frac 1 {\alpha} \frac{9\log n}{\eps^2}.
	 \end{align}
	 Let $S'$ be a subset of $V$. We define $X_{S'}= \sum_{u,v\in S'} X_{u,v} $, and $Y_{S'}= \sum_{u,v\in S'} Y_{u,v} $. Note that we have $X_{S'}\leq Y_{S'}$. We have 
	 	$\E[Y_{S'}]= \sum_{u,v\in S'} \E[Y_{u,v}] = \alpha \sum_{u,v\in S'} w_{u,v}$.
	 Hence, we have
	 \begin{align}\label{eq:dens:optE}
	 	\optd \geq \frac{\sum_{u,v\in S'} w_{u,v}}{|S'|} = \frac{\E[Y_{S'}]}{\alpha|S'|} \geq \frac{\E[X_{S'}]}{\alpha |S'|}
	 \end{align}
	  Note that $X_{u,v}$'s are chosen independently, and hence by applying the Chernoff bound to $X_{S'}$ for $\epsilon = \eps \frac{\alpha\optd|S'|}{\E[X_{S'}]}$ we have 
	 \begin{align*}
	 		\prob{|Y_{S'}-\E[Y_{S'}]| \geq \eps \alpha \optd |S'| } &= 
	 		\prob{|X_{S'}-\E[X_{S'}]| \geq \eps \alpha \optd |S'| } &\text{If $Y_{e}\neq X_e$, $Y_{e}=\E[Y_{e}]$}\\	 		
	 		&\leq 2 \exp\Big(-\frac 1 3\big(\eps \frac{\alpha\optd|S'|}{\E[X_{S'}]}\big)^2\E[X_{S'}]\Big) &\text{Chernoff bound}\\
	 		& = 2 \exp\Big(-\frac 1 3\eps^2 \frac{\alpha^2{\optd}^2|S'|^2}{ \E[X_{S'}]}\Big)\\
	 		& \leq 2 \exp\Big(-\frac 1 3 \eps^2 \alpha \optd |S'| \Big) &\text{By Inequality \ref{eq:dens:optE}}\\
	 		& \leq 2 \exp\Big(-\frac 1 3 \eps^2 \frac{9\log n}{\eps^2} |S'| \Big) &\text{By Inequality \ref{eq:dens:optmin}}\\
	 		& = 2 \exp\big(- 3 |S'| \log n \big). \\
	 \end{align*}
	 Next we union bound over all choices of $S'$.
	 \begin{align*}
		\prob{\exists_{S'} \big|Y_{S'}-\E[Y_{S'}]\big| \geq \eps \alpha \optd |S'| } & =  \prob{\exists_ k \exists_{|S'|=k} \big|Y_{S'}-\E[Y_{S'}]\big| \geq \eps \alpha \optd k }\\
		& \leq \sum_{k=2}^{n} \prob{\exists_{|S'|=k} \big|Y_{S'}-\E[Y_{S'}]\big| \geq \eps \alpha \optd k } &\text{Union bound}\\
	 	& \leq \sum_{k=2}^{n} \sum_{|S'|=k} \prob{ \big|Y_{S'}-\E[Y_{S'}]\big| \geq \eps \alpha \optd k } &\text{Union bound}\\
	 	&\leq \sum_{k=2}^{n} \sum_{|S'|=k} 2 \exp\big(-3 k \log n \big) \\
	 	& = \sum_{k=2}^{n} 2 {n \choose k}  \exp\big(-3 k \log n \big) \\
	 	& \leq \sum_{k=2}^{n} 2 \exp\big(-3 k \log n  + k \log n \big) & {n \choose k} \leq n^k\\
	 	& \leq \sum_{k=2}^{n} 2 \exp\big(- 4 \log n  \big) & k \geq 2 \\
	 	& \leq 2 \exp\big( - 3 \log n \big) \\
	 	& = \frac 2 {n^3} < \frac 1 n. & n \geq 2
	 \end{align*}
	 Therefore, with probability at least $1-\frac 1 n$ simultaneously for all $S' \subseteq V$ we have 
	 \begin{align}\label{eq:dens:main}
	 	\big|Y_{S'}-\E[Y_{S'}]\big| \leq \eps \alpha \optd |S'|.
	 \end{align}
	Next we prove the statement of the theorem in the cases where Inequality \ref{eq:dens:main} holds. Let $S^*$ be a densest subgraph of $G$. We have 
	\begin{align*}
		 \frac {\sum_{u,v\in S} w_{u,v}}{|S|} &= \frac{\frac 1 {\alpha} \E[\sum_{u,v\in S} Y_{u,v}]}{|S|} & \text{$\E[Y_{u,v}]= \alpha w_{u,v}$}\\
		 &= \frac 1 {\alpha} \frac{\E[Y_{S}]}{|S|} & \text{Definition of $Y_{S}$}\\
		 &\geq \frac 1 {\alpha}\frac{ Y_{S}}{|S|} -  \eps \optd & \text{By Inequality \ref{eq:dens:main}}\\
		 & \geq \frac 1 {\alpha}\phi \max_{S''}\frac{ Y_{S''}}{|S''|} -  \eps \optd & \text{$S$ is a $\phi$ approximation on $\H$}\\
		 & \geq \frac 1 {\alpha}\phi \frac{ Y_{S^*}}{|S^*|} -  \eps \optd \\
		  & \geq \frac 1 {\alpha}\phi \frac{ \E[ Y_{S^*} ]}{|S^*|} -  2 \eps \optd &\text{By Inequality \ref{eq:dens:main}} \\
		  & \geq \phi \frac{ \sum_{u,v\in S^*} w_{u,v} }{|S^*|} -  2 \eps \optd & \text{$\E[Y_{u,v}]= \alpha w_{u,v}$}\\
		  & = (\phi - 2\eps)\optd. & \text{Definition of $S^*$}
	\end{align*}
	  
\end{proof}

Recall that, as stated in the introduction, this result implies a $(1/2-\epsilon)$-approximation algorithm for densest subgraph in $\lambda$-metric spaces requiring $\tilde{O}(\frac{n}{\lambda \epsilon^2})$ time.

The following theorem shows the efficiency of our technique for $k$-hypermatching.

\begin{theorem}
	Choose $\i = \frac{6\log n}{\eps^2} \frac{n^2}{k-1} \in \tilde{O}\big(\frac {n^2} {\epsilon^2k}\big)$. Let $\S= \{S_1,S_2,\dots,S_{n/k}\}$ be a $\phi$-approximation solution to the $k$-hypermatching on unweighted graph $\H$. $\S$ is a $\phi-2\eps$ approximation solution to the $k$-hypermatching on $G$, with probability at least $1-\frac 1 n$.
\end{theorem}
\begin{proof}
	Let $\S''= \{S''_1,S''_2,\dots,S''_{n/k}\}$ be a $k$-hypermatching chosen uniformly at random among all $k$-hypermatchings. Note that the number of edges that fall in $\S''$ is $\frac n k {k \choose 2} = \frac{n(k-1)}{2}$, while there are ${n \choose 2}= \frac {n(n-1)}{2}$ edges in $G$ in total. Hence, due to symmetry each edge falls in $\S''$ with probability $\frac {k-1}{n-1} \leq \frac{k-1}{n}$. Now, we give a lower bound on $\optm$. We later use this bound in our concentration bound.
	\begin{align}\label{eq:match:optmin}
	\optm \geq \E [\sum_{i=1}^{n/k} \sum_{u,v \in S''_i} w_{u,v}] = \frac{k-1}{n} \sum_{u,v \in V} w_{u,v}  =  \frac{k-1}{n} \frac 1 {\alpha} \sum_{u,v\in V} \E[Y_{u,v}] \geq \frac{k-1}{n} \frac 1 {\alpha} \beta = \frac 1 {\alpha} \frac{6n\log n}{\eps^2}.
	\end{align}
	Let $\S'= \{S'_1,S'_2,\dots,S'_{n/k}\}$ be a $k$-hypermatching of $G$ (i.e., a decomposition of $V$ into $n/k$ distinct subsets of size $k$). We define $X_{\S'}= \sum_{i=1}^{n/k} \sum_{u,v \in S'_i} X_{u,v} $ and $Y_{\S'}= \sum_{i=1}^{n/k} \sum_{u,v \in S'_i} Y_{u,v} $. We have 
	$\E[X_{\S'}]= \sum_{i=1}^{n/k} \sum_{u,v \in S'_i} \E[X_{u,v}] = \alpha \sum_{i=1}^{n/k} \sum_{u,v \in S'_i} w_{u,v}$.
	Hence we have
	\begin{align}\label{eq:match:optE}
	\optm \geq \sum_{i=1}^{n/k} \sum_{u,v \in S'_i} w_{u,v} = \sum_{i=1}^{n/k} \sum_{u,v \in S'_i} \frac 1 {\alpha}\E[Y_{\S'}] \geq \sum_{i=1}^{n/k} \sum_{u,v \in S'_i} \frac 1 {\alpha}\E[X_{\S'}].
	\end{align}
	Note that $X_{u,v}$'s are chosen independently, and hence by applying the Chernoff bound to $X_{\S'}$ for $\epsilon = \eps \frac{\alpha\optm}{\E[X_{\S'}]}$ we have 
	\begin{align*}
	\prob{|Y_{\S'}-\E[Y_{\S'}]| \geq \eps \alpha \optm } &= 
	\prob{|X_{\S'}-\E[X_{\S'}]| \geq \eps \alpha \optm }  &\text{If $Y_e \neq X_e$, $Y_e=\E[Y_e]$}\\
	&\leq 2 \exp\Big(-\frac 1 3\big(\eps \frac{\alpha\optm}{\E[X_{\S'}]}\big)^2\E[X_{\S'}]\Big) &\text{Chernoff bound}\\
	& = 2 \exp\Big(-\frac 1 3\eps^2 \frac{\alpha^2{\optm}^2}{ \E[X_{\S'}]}\Big)\\
	& \leq 2 \exp\Big(-\frac 1 3 \eps^2 \alpha \optm  \Big) &\text{By Inequality \ref{eq:match:optE}}\\
	& \leq 2 \exp\Big(-\frac 1 3 \eps^2 \frac{6n\log n}{\eps^2}  \Big) &\text{By Inequality \ref{eq:match:optmin}}\\
	& = 2 \exp\big(- 2 n \log n \big) \\
	\end{align*}
	
	Next we union bound over all choices of $\S'$.
	\begin{align*}
	\prob{\exists_{\S'} \big|Y_{\S'}-\E[Y_{\S'}]\big| \geq \eps \alpha \optm } & \leq \sum_{\S'} \prob{ \big|Y_{\S'}-\E[Y_{\S'}]\big| \geq \eps \alpha \optm } &\text{Union bound}\\
	&\leq \sum_{\S'} 2  \exp\big(-2 n \log n \big) \\
	& \leq 2 n^n \exp\big(-2 n \log n \big) \\
	& \leq 2 \exp\big( - n\log n \big) \leq \frac 1 n. & n \geq 2
	\end{align*}
	Therefore, with probability $1-\frac 1 n$ simultaneously for all $\S' \subseteq V$ we have 
	\begin{align}\label{eq:match:main}
	\big|Y_{\S'}-\E[Y_{\S'}]\big| \leq \eps \alpha \optm.
	\end{align}
	Next we prove the statement of the theorem in the cases where Inequality \ref{eq:match:main} holds. Let $\S^*= \{S_1^*,S_2^*,\dots,S_{n/k}^*\}$ be a maximum $k$-hypermatching of $G$. We have 
	\begin{align*}
	 {\sum_{i=1}^{n/k}\sum_{u,v\in S_i} w_{u,v}} &= {\frac 1 {\alpha} \E\Big[\sum_{i=1}^{n/k}\sum_{u,v\in S_i} Y_{u,v}\Big]} & \text{$\E[Y_{u,v}]= \alpha w_{u,v}$}\\
	&= \frac {\E[Y_{\S}]} {\alpha}  & \text{Definition of $Y_{\S}$}\\
	&\geq \frac { Y_{\S}} {\alpha} -  \eps \optm & \text{By Inequality \ref{eq:match:main}}\\
	& \geq \phi\frac{ \max_{\S''} Y_{\S''}} {\alpha}  -  \eps \optm & \text{$\S$ is a $\phi$ approximation on $\H$}\\
	& \geq \phi\frac{ Y_{\S^*}} {\alpha}  -  \eps \optm & \\
	& \geq \phi\frac{ \E[ Y_{\S^*} ]} {\alpha}  -  2 \eps \optm &\text{By Inequality \ref{eq:match:main}} \\
	& \geq \phi {\sum_{i=1}^{n/k} \sum_{u,v\in \S^*_i} w_{u,v} } -  2 \eps \optm & \text{$\E[Y_{u,v}]= \alpha w_{u,v}$}\\
	& = (\phi - 2\eps)\optm. & \text{Definition of $S^*$}
	\end{align*}
	
\end{proof}

Finally we show the efficiency of our sketch for finding the maximum cut, again following the same basic proof outline. Here, we indicate a cut by the set of vertices of its smaller side, breaking ties arbitrarily.

\begin{theorem}
	Choose $\i = \frac{18\log n}{\eps^2} n$. Let $S$ be a $\phi$-approximation solution to the maximum cut on $\H$. $S$ is a $\phi-2\eps$ approximation solution to the maximum cut on $G$, with probability at least $1-\frac 1 n$.
\end{theorem}
\begin{proof}
	First we lower bound $\optc$. Note that in the optimum solution moving a vertex from one side to the other does not increase the value of the cut. Thus, for each vertex $v\in V$ the total weight of the edges neighboring $v$ in the cut is at least half of the total weight of all edges neighboring $v$. Hence we have 
	\begin{align}\label{eq:cut:optmin}
	\optc \geq \frac 1 2 \sum_{v\in V} \sum_{u \in V}  w_{u,v} =  \frac{1}{2 \alpha}\sum_{v\in V} \sum_{u \in V} \E[Y_{u,v}] \geq \frac{1}{2\alpha} \beta = \frac 1 {\alpha} \frac{9\log n}{\eps^2}n.
	\end{align}
	Let $S'$ be a subset of $V$. We define $X_{S'}= \sum_{v\in S'}\sum_{u\in V\setminus S'} X_{u,v} $, and $Y_{S'}= \sum_{v\in S'}\sum_{u\in V\setminus S'} Y_{u,v} $. Note that we have $X_{S'}\leq Y_{S'}$. We have 
	$\E[Y_{S'}]= \sum_{v\in S'}\sum_{u\in V\setminus S'} \E[Y_{u,v}] = \alpha \sum_{v\in S'}\sum_{u\in V\setminus S'} w_{u,v}$.
	Hence, we have
	\begin{align}\label{eq:cut:optE}
	\optc \geq \sum_{v\in S'}\sum_{u\in V\setminus S'} w_{u,v} = \frac{\E[Y_{S'}]}{\alpha} \geq \frac{\E[X_{S'}]}{\alpha }. 
	\end{align}
	Note that the $X_{u,v}$'s are independent, and hence by applying the Chernoff bound to $X_{S'}$ for $\epsilon = \eps \frac{\alpha\optc}{\E[X_{S'}]}$ we have 
	\begin{align*}
	\prob{|Y_{S'}-\E[Y_{S'}]| \geq \eps \alpha \optc } &= 
	\prob{|X_{S'}-\E[X_{S'}]| \geq \eps \alpha \optc } &\text{If $Y_{e}\neq X_e$, $Y_{e}=\E[Y_{e}]$}\\	 		
	&\leq 2 \exp\Big(-\frac 1 3\big(\eps \frac{\alpha\optc}{\E[X_{S'}]}\big)^2\E[X_{S'}]\Big) &\text{Chernoff bound}\\
	& = 2 \exp\Big(-\frac 1 3\eps^2 \frac{\alpha^2{\optc}^2}{ \E[X_{S'}]}\Big)\\
	& \leq 2 \exp\Big(-\frac 1 3 \eps^2 \alpha \optc  \Big) &\text{By Inequality \ref{eq:cut:optE}}\\
	& \leq 2 \exp\Big(-\frac 1 3 \eps^2 \frac{9\log n}{\eps^2} n \Big) &\text{By Inequality \ref{eq:cut:optmin}}\\
	& = 2 \exp\big(- 3 n \log n \big) \\
	\end{align*}
	Next we union bound over all choices of $S'$.

	\begin{align*}
	\prob{\exists_{S'} \big|Y_{S'}-\E[Y_{S'}]\big| \geq \eps \alpha \optc } &
	\leq \sum_{S'\subseteq V} \prob{ \big|Y_{S'}-\E[Y_{S'}]\big| \geq \eps \alpha \optc k } &\text{Union bound}\\
	&\leq \sum_{S'\subseteq V} 2 \exp\big(-3 n \log n \big) \\
	& = 2^{n+1} \exp\big(-3 n \log n \big) \\
	& \leq \frac 1 {n^n} < \frac 1 n. & n \geq 2
	\end{align*}
	Therefore, with probability at least $1-\frac 1 n$ simultaneously for all $S' \subseteq V$ we have 
	\begin{align}\label{eq:cut:main}
	\big|Y_{S'}-\E[Y_{S'}]\big| \leq \eps \alpha \optc.
	\end{align}
	Next we prove the statement of the theorem in the cases where Inequality \ref{eq:cut:main} holds. Let $S^*$ be a maximum cut of $G$. We have 
	\begin{align*}
	\sum_{v\in S^*}\sum_{u\in V\setminus S^*} w_{u,v} &= \sum_{v\in S^*}\sum_{u\in V\setminus S^*}\frac 1 {\alpha} \E[Y_{u,v}] & \text{ $\E[Y_{u,v}]= \alpha w_{u,v}$}\\
	&= \frac 1 {\alpha} \E[Y_{S}] & \text{Definition of $Y_{S}$}\\
	&\geq \frac 1 {\alpha} Y_{S} -  \eps \optc & \text{By Inequality \ref{eq:cut:main}}\\
	& \geq \frac 1 {\alpha}\phi \max_{S''} Y_{S''} -  \eps \optc & \text{$S$ is a $\phi$ approximation on $\H$}\\
	& \geq \frac 1 {\alpha}\phi  Y_{S^*} -  \eps \optc\\
	& \geq \frac 1 {\alpha}\phi  \E[ Y_{S^*} ] -  2 \eps \optc &\text{By Inequality \ref{eq:cut:main}} \\
	& \geq \phi { \sum_{v\in S^*}\sum_{u\in V\setminus S^*} w_{u,v} } -  2 \eps \optc & \text{$\E[Y_{u,v}]= \alpha w_{u,v}$}\\
	& = (\phi - 2\eps)\optc. & \text{Definition of $S^*$}
	\end{align*}
	
\end{proof}

\section{Impossibility Results}\label{sec:hard}
In this section we consider all of the problems of the previous section and show that it is necessary to use $\Omega(n)$ queries even if we just want to estimate the value of the solutions. In particular, we show that $\Omega(n)$ queries are required to distinguish the following two graphs.
\begin{itemize}
	\item In $G_1$ we have $n$ vertices $\{v_1,\dots,v_n\}$ and the weight of all edges are $0$.
	\item In $G_2$ again we have $n$ vertices. Pick an index $r \in \{1,\dots,n\}$ uniformly at random. The weight of each edge neighboring $v_r$ is  $1$. The weight of all other edges is $0$.
\end{itemize}

The following lemma shows the hardness of distinguishing $G_1$ and $G_2$.

\begin{lemma} \label{lm:hard:G1G2}
	For any $\delta \in (0,0.5]$, it is impossible to distinguish $G_1$ and $G_2$ using $\delta n - 1$ queries with probability $0.5+\delta$.
\end{lemma}
\begin{proof}
	Let $\alg$ be a (possibly randomized) algorithm that distinguishes $G_1$ and $G_2$ using at most $\delta n - 1$ queries. For simplicity, and without loss of generality, we assume that $\alg$ makes exactly $\delta n -1$ queries. Let $(u_1,u_2),(u_3,u_4), \dots , (u_{k-1},u_{k})$ be the sequence of edges probed by $\alg$, where the $u_i$'s may be random variables and $k=2\delta n-2 $. Notice that $v_r$ is chosen uniformly at random. Hence, in case that the input is $G_2$, for any arbitrary $j\in\{1,\dots,k\}$ we have $\prob{u_j = v_r} = \frac 1 n$. Therefore, we have
	\begin{align*}
		\prob{\exists_{i\in \{1,\dots,\frac k 2\}} w_{u_{2i-1},u_{2i}}\neq 1} &= 
		\prob{\exists_{j \in \{1,\dots,k\} } u_j = v_r }   \\
		&\leq \sum_{j=1}^{k} \prob{u_i = v_r} &\text{By union bound}\\
		&= \frac k n  &\text{Since $\prob{u_j = v_r} = \frac 1 n$}\\
		&< 2\delta.  &\text{Since $k=2\delta n-2 $}
	\end{align*}
	Hence, in the case that the input is $G_2$, with probability at least $1-2\delta$ all the edges that $\alg$ queries have weight $0$.  Trivially, in the case that the input is $G_1$ all the queried edges have weight $0$. Therefore the probability that $\alg$ distinguishes $G_1$ and $G_2$ is less than $2\delta + \frac{1-2\delta}{2} = 0.5+\delta$.
\end{proof}

Note that the weight of the edges of $G_1$ is $0$, while average weight of the edges of $G_2$ is $\frac {n-1}{{n \choose 2}} = \frac{2}{n}$.   Therefore any algorithm that estimates the average weight of the edges within any multiplicative factor distinguishes $G_1$ and $G_2$. This together with Lemma \ref{lm:hard:G1G2} proves the following corollary.

\begin{corollary}
	Any approximation algorithm that estimates the average distance a in metric graphs within any multiplicative factor with probability $0.51$ requires $\Omega(n)$ queries.
\end{corollary}

Note that the density of the densest subgraph of $G_1$ is $0$, while the density of the densest subgraph of $G_2$ is $\frac {n-1}{n} \geq \frac{1}{2}$. Therefore any algorithm that estimates the density of the densest subgraph within any multiplicative factor distinguishes $G_1$ and $G_2$. This together with Lemma \ref{lm:hard:G1G2} proves the following corollary.

\begin{corollary}
	Any approximation algorithm that estimates the density of the densest subgraph in a metric graphs within any multiplicative factor with probability $0.51$ requires $\Omega(n)$ queries.
\end{corollary}

Notice that the value of the maximum matching of $G_1$ is $0$ while the value of the maximum matching of $G_2$ is $1$. Therefore any algorithm that estimates the value of the maximum matching within any multiplicative factor distinguishes $G_1$ and $G_2$. This together with Lemma \ref{lm:hard:G1G2} proves the following corollary.

\begin{corollary}
	Any approximation algorithm that estimates the value of the maximum matching in a metric graphs within any multiplicative factor with probability $0.51$ requires $\Omega(n)$ queries.
\end{corollary}

Notice that the value of the maximum cut of $G_1$ is $0$ while the value of the maximum cut of $G_2$ is $n-1$. Therefore any algorithm that estimates the value of the maximum cut distinguishes $G_1$ and $G_2$. This together with Lemma \ref{lm:hard:G1G2} proves the following corollary.

\begin{corollary}
	Any approximation algorithm that estimates the value of the maximum cut in a metric graphs within any multiplicative factor with probability $0.51$ requires $\Omega(n)$ queries.
\end{corollary}

\section{Conclusion}
We have show that in metric graphs one can efficiently obtain a linear
sampling with a sublinear number of edge queries, allowing efficient
sparsification that leads to efficient approximation algorithms.  We
believe this technique may be useful in generating approximation
algorithms for other problems beyond those considered here.  Open questions
include possibly improving the lower bounds, or otherwise bridging the gap
between the upper and lower bounds on required queries.

\bibliographystyle{plain}

\newpage
\appendix
\section{Uniform edge sampling  fails to find the densest subgraph}\label{app:dens}
It is known that for general unweighted graphs, if we sample each edge with a small probability $p \in \tilde{\Omega}(\frac{1}{\epsilon^2 n})$, the densest subgraph of the sampled subgraph is a $(1-\epsilon)$-approximation of the densest subgraph of the original graph \cite{esfandiari2015applications}. Here with a simple example we show that this result is not true for weighted graphs in a metric space even when $p$ is a small constant.

Consider a graph $G$ with vertex set $V=\{v_1,v_2,\dots,v_n\}$, where the weight of each each intersecting $v_1$ is $\frac n 2 + 1$ and the weight of each other edge is $1$. The densest subgraph of $G$ contains the whole graph, and its density is $\frac{{n-1 \choose 2}+(n-1)(\frac n 2+1)}{n} = \frac{(n-1)n}{n} = n-1$.

Let $G_p$ be a subgraph of $G$ obtained by sampling each edge with probability $p$. Using a simple Chernoff bound it is easy to show that with high probability $G_p$ has at most $2p{n-1\choose 2}$ edges of weight $1$. Similarly, with high probability the number of edges of weight $\frac n 2+1$ in $G_p$ is between $1$ and $2p(n-1)$. 

Let $H=(V_H,E_H)$ be the densest subgraph of $G_p$. We have
\begin{align*}
	\frac{\sum_{e\in E_H} w_e}{|V_H|} \leq \frac{2p{n-1\choose 2} + 2p(n-1) (n/2 +1)}{|V_H|} = \frac{2pn (n-1)}{|V_H|}.
\end{align*}
On the other hand the density of one single edge with weight $\frac n 2+1$ is $\frac n 4+\frac 1 2$. Thus we have $\frac n 4+\frac 1 2 \leq \frac{2pn (n-1)}{|V_H|}$ which implies $|V_H|\leq 8 p n$. Therefore the density of the densest subgraph induced by $|V_H|$ is at most 
\begin{align*}
\frac{{|V_H|-1 \choose 2} + |V_H| (\frac n 2 +1)}{|V_H|}
\leq \frac {|V_H|-1}{2} + \frac n 2 +1 \leq 4pn - \frac 1 2 +\frac n 2 +1 = (\frac 1 2 + 4p)n+ \frac 1 2.
\end{align*}
Therefore, the subgraph of $G$ induced by $V_H$ is not better than a $\frac{(\frac 1 2 + 4p)n+ \frac 1 2}{n-1}\simeq 0.5+4p$ approximate solution.

\end{document}